\documentclass[11pt]{article}
\usepackage{booktabs} 
\usepackage[ruled]{algorithm2e} 

\SetAlFnt{\small}
\SetAlCapFnt{\small}
\SetAlCapNameFnt{\small}
\SetAlCapHSkip{0pt}
\usepackage{authblk}

\usepackage{amsmath,amsfonts,amsthm}
\usepackage{thmtools,thm-restate} 
\usepackage{hyperref}
\usepackage{xcolor}
\usepackage{times}
\usepackage{amssymb}
\usepackage{bm}
\usepackage{xspace}
\usepackage{ulem}
\usepackage{mathtools}
\usepackage{wrapfig}

\usepackage{nicefrac}
\usepackage{natbib}

\usepackage{comment}

\usepackage{cleveref}
\usepackage{caption}
\usepackage{graphicx}
\usepackage{subcaption}
\usepackage{enumitem}

\usepackage{ulem}

\usepackage[margin=1in]{geometry}

\Crefname{equation}{Eq.}{Eqs.}
\Crefname{figure}{Fig.}{Fig.}
\Crefname{tabular}{Tab.}{Tabs.}
\Crefname{table}{Tab.}{Tabs.}
\Crefname{section}{Section}{Section}
\Crefname{appendix}{App.}{App.}

\newtheorem{theorem}{Theorem}
\newtheorem{lemma}{Lemma}
\newtheorem{proposition}{Proposition}
\newtheorem{claim}{Claim}
\newtheorem{corollary}{Corollary}
\newtheorem{definition}{Definition}

\newtheorem{fact}{Fact}

\newtheorem{remark}{Remark}
\newtheorem{assumption}{Assumption}

\def\R{\mathbb{R}}
\def \P{\mathbb{P}}

\def\vp{{\bm{p}}}

\def\vv{{\bm{v}}}

\def\vx{{\bm{x}}}

\newcommand{\E}{\mathbb{E}}

\DeclarePairedDelimiter{\ceil}{\lceil}{\rceil}

\DeclarePairedDelimiter{\set}{ \{ }{ \} }

\ifdefined\NOCOMMENT
    \newcommand{\ggnote}[1]{}
    \newcommand{\apnote}[1]{}
    \newcommand{\knote}[1]{}
    \newcommand{\ganote}[1]{}
    
\else
    \newcounter{note}[section]
    \renewcommand{\thenote}{\thesection.\arabic{note}}
    \newcommand{\ggnote}[1]{\refstepcounter{note}\textcolor{blue}{$\ll${\bf Guru~\thenote:} {\sf #1}$\gg$\marginpar{\tiny\bf GG~\thenote}}}
    
    \newcommand{\apnote}[1]{\refstepcounter{note}\textcolor{violet}{$\ll${\bf Andres~\thenote:} {\sf #1}$\gg$\marginpar{\tiny\bf AP~\thenote}}}
     
    \newcommand{\knote}[1]{\refstepcounter{note}\textcolor{red}{$\ll${\bf Kshipra~\thenote:} {\sf #1}$\gg$\marginpar{\tiny\bf KBL~\thenote}}}

    \newcommand{\ganote}[1]{\refstepcounter{note}\textcolor{purple}{$\ll${\bf Gagan~\thenote:} {\sf #1}$\gg$\marginpar{\tiny\bf GA~\thenote}}}
\fi

\newcommand{\ostat}[2]{\ensuremath{v^{( #1, #2) }}}

\newcommand{\robustmechanism}{Intermediary-Proof Mechanism \xspace}
\newcommand{\heteromechanism}{Intermediary-Proof Mechanism for Heterogeneous Items \xspace}
\newcommand{\robustmechanismshort}{IPM\xspace} 
\newcommand{\bu}{\ensuremath{i}}
\newcommand{\ag}{\ensuremath{l}}
\newcommand{\ob}{\ensuremath{j}}

\newcommand{\discon}{Buyer-Disconnectability\xspace}

\title{Maximizing revenue in the presence of intermediaries}
\author[]{Gagan Aggarwal}
\author[]{Kshipra Bhawalkar}
\author[]{Guru Guruganesh}
\author[]{Andres Perlroth}
\affil[]{Google Research}

\date{}

\begin{document}

\maketitle

\begin{abstract}
We study the mechanism design problem of selling $k$ items to unit-demand buyers with private valuations for the items. A buyer either participates directly in the auction or is represented by an intermediary, who represents a subset of buyers. Our goal is to design robust mechanisms that are independent of the demand structure (i.e. how the buyers are partitioned across intermediaries), and perform well under a wide variety of possible contracts between intermediaries and buyers.

We first consider the case of $k$ identical items where each buyer draws its private valuation for an item i.i.d. from a known $\lambda$-regular distribution. We construct a robust mechanism that, independent of the demand structure and under certain conditions on the contracts between intermediaries and buyers, obtains a constant factor of the revenue that the mechanism designer could obtain had she known the buyers' valuations. In other words, our mechanism's expected revenue achieves a constant factor of the optimal welfare, regardless of the demand structure. Our mechanism is a simple posted-price mechanism that sets a take-it-or-leave-it per-item price that depends on $k$ and the total number of buyers, but does not depend on the demand structure or the downstream contracts.

Next we generalize our result to the case when the items are not identical. We assume that the item valuations are separable, i.e. $v_{\bu \ob} = \eta_\ob v_\bu$ for buyer $\bu$ and item $\ob$, with each private $v_\bu$ drawn i.i.d. from a known $\lambda$-regular distribution. For this case, we design a mechanism that obtains at least a constant fraction of the optimal welfare, by using a menu of posted prices. This mechanism is also independent of the demand structure, but makes a relatively stronger assumption on the contracts between intermediaries and buyers, namely that each intermediary prefers outcomes with a higher sum of utilities of the subset of buyers represented by it.
\end{abstract}
%
\thispagestyle{empty}
\newpage
\setcounter{page}{1}

 \section{Introduction}
\normalem
The fact that {\em the game is always bigger than you think} is a major concern in applied mechanism design. In many applications, auctioneers design their auction assuming that the ultimate consumers are the ones bidding in the auction. However, in many markets, buyers are not bidding in the auction but are, instead, being represented by intermediaries. 
For example, intermediaries are a common, though often ignored, element of advertisement auctions on the Internet. In Sponsored search and Display ad auctions, advertisers are often represented by ad agencies who optimize their online ad campaigns.\footnote{See \citet{choi2020online} for a comprehensive review of intermediation in ad-auctions.} Another example is the AdExchange marketplace that sells advertisements on webpages. In an AdExchange, the advertisers are represented by Ad Networks, who do real-time bidding on their behalf into the AdExchange auction \citep{muthukrishnan2009ad}. Similarly, intermediation plays an important role in over-the-counter markets. As described in \citet{duffie2005over}, customers interested in buying such assets usually execute their trades with dealers (intermediaries), which are the ones interacting with the seller.

In this context, classic auction design may under-perform since intermediaries may 
enable buyers to engage in collusive behavior, or because of the double marginalization issue impacting the  intermediaries' preferences for outcomes. The information about which buyers each intermediary represents and the 
intermediaries' incentives in the auction (which depend on the downstream negotiation with buyers) is usually hard for the seller to know. With this in mind, we design mechanisms that do not rely on this information. Specifically, we design an intermediary-proof mechanism with two properties:
\begin{enumerate}
    \item It is independent of the market structure -- which buyers are represented by which intermediaries.
    \item The revenue obtained is a constant factor of the optimal welfare (the highest revenue the auctioneer could possibly obtain).
\end{enumerate} 
Furthermore, our mechanism consists of a simple posted-pricing scheme that depends only on the number of buyers, number of items and some order statistics of the buyers' valuation distribution. 

\subsection{Model and main results}
Our model consists of one seller (the auctioneer) who owns $k$ items and $n$ unit-demand buyers. 
We start with the problem where the items are identical and later tackle the case of non-identical but related items. For the identical items case, each buyer has the same valuation for each item, which is unknown to the seller. The valuations $(v_\bu)_{\bu \in [n]}$ are assumed to be drawn independently from the same distribution $F$. We further assume that the distribution $F$ is $\lambda$-regular (see Definition~\ref{def:lambda-regular} for details); note that the class of $\lambda$-regular distributions allows us to smoothly interpolate between monotone hazard rate  (MHR) distributions ($\lambda=0$) and the more general class of regular distributions $(\lambda=1)$. Each buyer either participates directly in the auction or is represented by one of $m$ intermediaries. Taking an agnostic approach to how these decisions are made, we consider a general partition of the buyers where all the buyers in one part are represented by the same intermediary. Next, we abstract from the downstream decisions and contracts between buyers and intermediaries and consider a general class of utility functions for intermediaries. This class is motivated by, and includes, two of the most natural and commonly-studied intermediary utility functions: (a) surplus-maximizing intermediary which maximizes the surplus of the buyers it represents and (b) the monopolist intermediary that maximizes its profit from captive buyers. The class consists of all utility functions that satisfy three natural properties: that each intermediary can secure at least a payoff of zero (individual rationality); that their payoff cannot be greater than the surplus created by the buyers (i.e.~intermediaries don't create value); and that when the seller posts a uniform price $r$ per item such that there are $q$ Buyers represented by Intermediary $\ag$ with valuation at least $r$, then Intermediary $\ag$ will buy at least $\tau_\ag q$ items in expectation (\discon). 

The first result of the paper is the construction of a mechanism that for any partition of the buyers and any aggregators' utility functions that satisfy the above conditions, guarantees at least $\tau \frac{1}{2 }  (1-\lambda)^{\frac 1 \lambda} \left(1-\frac 1 e\right)$ fraction of the welfare where $\tau = \min \tau_\ag$ (see Theorem~\ref{th:1}). This translates to a $\tau \cdot 1/2e (1-1/e)$ approximation factor for MHR distributions and a $0$ approximation factor for regular distributions\footnote{We show that any mechanism that gets the same revenue from different market structures (like ours) cannot get a constant factor of the revenue for regular distributions -- see Section~\ref{sec:discussion} for details.}. In other words, independent of which buyers each intermediary represents and the exact details 
of the downstream interaction between buyers and intermediaries, 
our mechanism gets a constant-factor approximation of the revenue the seller could obtain if she had complete information about the buyers' valuations. We would like to note that a priori, it is not at all obvious that such a robust mechanism exists. As we show in \Cref{sec:robustness}, implementing the revenue-maximizing mechanism for a particular partition of buyers can perform poorly when faced with a different partition of buyers. The mechanism is a uniform posted-pricing scheme that only depends on the number of items $k$, number of bidders $n$, and the distribution $F$. More precisely, the seller posts a price per item $p^{R} =\E [\ostat{1}{ \lceil n/k \rceil}]$, where $\ostat{1}{\ob}$ is the largest order statistic of $\ob$ samples drawn from $F$.\footnote{In case of excess demand, items are randomly allocated among interested bidders.}



Next in Section~\ref{sec:ext}, we consider the problem when the items are not identical. In particular, we consider the case of separable values, i.e. Buyer $\bu$'s valuation of receiving item $\ob$ is given by $\eta_\ob v_\bu$. This is inspired by sponsored search advertising, where different positions on the page may have different click-through-rates and therefore different values (note that this does not model sponsored search exactly as in our model, we can sell an arbitrary subset of items, not necessarily a prefix). Compared to the identical items case, the difficulty here is that to compete with welfare, higher-weighted items must be allocated to high-valuation bidders. This requires the mechanism to charge a higher cost for higher-weight items and screen buyers according to their preference for high-value/high-cost items. To do this, we make a more detailed assumption about the intermediaries' preferences and consider the case where each intermediary is a surplus-maximizing intermediary. With this, we construct a sequential posted-price mechanism that, for every partition of buyers, obtains revenue equal to a constant fraction of the welfare. Although the pricing scheme is more subtle than for the homogeneous case, it nevertheless still possesses the nice property of depending only on the number of items $k$, the weight values $(\eta_\ob)_{\ob \in [k]}$, the number of bidders $n$, and the valuation distribution $F$.


\subsection{Why is this interesting?}

From an {\bf applied perspective}, our intermediary-proof 
mechanisms have appealing properties which, aside from the revenue guarantees, make them very suitable for applications. First, their simplicity. Our mechanisms are a posted pricing, requiring a simple menu to implement --
this makes bidders less likely to make strategic errors \citep{shengwuli2017}. This contrasts with optimal auctions in multidimensional environments (like ours) which are usually quite complex \citep{daskalakis2017strong}. Second, the mechanisms are detail-free -- the seller only requires knowledge of the number of buyers and good estimates about buyers' valuations \footnote{In fact, our mechanism only requires us to estimate the order statistics of the buyers distributions.}, and does not need to learn any information about the market structure or the details of the intermediaries' utility functions to implement the mechanism. 

In fact, the price that we use, $\E [\ostat{1}{ \lceil n/k \rceil}]$, where $\ostat{1}{\ob}$ is the largest order statistic of $\ob$ samples drawn from $F$, is a very natural function of the valuation distribution. Intuitively, this captures the {\em real} competition for each item since there are $n/k$ buyers per item. Although pricing based on the order statistics of the buyers' distributions are commonly used in mechanism design, we believe this particular function is new and its generalizations might prove useful in other multi-item settings.


A consequence of our results (see \Cref{coro:1}), that might be of independent interest, is that the optimal revenue for every demand structure and bidders' utility functions (that satisfy the above conditions) is at least a constant fraction of the optimal welfare (which is the same for all demand structures) assuming that the distribution is $\lambda$-regular with $\lambda$ bounded away from $1$ by a constant.
In particular, this is true for the full competition case, when all buyers participate directly in the auction and their utility functions are simply quasi-linear. Consider any per-item uniform posted-price mechanism (for the identical items case) that is constant-competitive with optimal revenue under full competition. Since a uniform posted-price mechanism achieves the same revenue under different demand structures (upto a factor of $\tau$), this mechanism would also be constant-competitive with respect to optimal welfare (and thus optimal revenue) under all demand structures and bidders' utility functions satisfying the above conditions.

Another result that might be of independent interest is Lemma~\ref{lem:1}, which shows that the largest order statistic of a $\lambda$-regular distribution is also distributed according to a $\lambda$-regular distribution. The analogous property for MHR distributions has proved quite useful in analyzing mechanisms for those distributions, and we believe this property might be helpful in extending some of those results to the $\lambda$-regular case.


 \subsection{Related Work}

Our work differs from the literature about intermediation in auctions by taking a robust approach to the market structure and the downstream contracts between buyers and intermediaries. Prior work has tried to solve the optimal mechanism for specific market structures while our mechanism provides a constant approximation. Because finding the optimal mechanism has the technical difficulty of first the need to characterize the subgame between aggregators and buyers, most of this literature has centered their attention in describing the downstream dynamics with a restrictive assumption on the number of items, intermediaries, and buyers \citep{hummel2016incentivizing, stavrogiannis2013competing, balseiro2019intermediation, balseiro2017optimal}.\footnote{An important exception is \citet{allouah2017auctions} who shows in a general framework that it is Pareto improving not to restrict the number of bids each intermediary can submit.} 
Closer to our work, \citet{feldman2010auctions} shows that for the one-item case, the optimal mechanism when each intermediary has one buyer captive consists of a second price auction (SPA) with a reserve price that is decreasing with the number of bidders. Our mechanism instead proposes a fixed pricing, which increases with the number of bidders.\footnote{\citet{feldman2010auctions} also restricts attention to valuations distribution satisfying the monotone hazard rate property, and hence, our revenue guarantee apply to their setting.} Finally, when buyers' valuations are fully correlated, \citet{bergemann2017selling} shows that the optimal mechanism is a posted-price mechanism.

\subsubsection{Collusion in Auctions}
The second stream of literature related to our work is the one studying collusion in auctions (see \citet{krishna2009auction} for a textbook review of this topic). This area studies auction design when bidders form bidding rings (cartel) and collude with other bidders in the ring. In this context, a bidding ring can be seen as an intermediary in our model. \citet{che2006robustly} provide a methodology to make a mechanism robust to any form of collusion. \citet{che2009optimal,pavlov2008auction} provide conditions where they can reshape the optimal mechanism into a collusion-proof one. Although such results go along the lines of our robustness approach regarding the market structure (here, how the cartels are conformed), such mechanisms differ from ours in two crucial layers. (i) They depend on the bidders' specific payoff structures, while in our mechanisms, such dependence is minimal. (ii) Their mechanisms are not ex-post individually rational while our robust mechanisms satisfy this property.\footnote{\citet{che2018weak} shows that under some restriction on the cartels, the optimal collusion-proof mechanism satisfies ex-post individually rationality.}  


\subsubsection{Multi-parameter Revenue-Optimal Auctions}
Our work is also closely related to the work on multi-parameter mechanism design. With intermediaries controlling multiple bids, they resemble buyers buying multiple goods with set function valuation over the set of goods.  In particular, when the utility function of each intermediary is the aggregate surplus of the buyers it represents, the intermediaries' value for a set of $i$ items is the sum of its $i$ highest values. This valuation function is submodular as adding an item to a smaller set yields higher incremental value. Note however that the intermediaries' values for disjoint subsets of items are not independent. 

There is a long line of research studying multi-parameter mechanism design. \cite{CaiDW12a, BriestCKW10, AlaeiFHHM12} show how to construct revenue-optimal mechanisms under fairly general conditions. However, these mechanisms are necessarily complex~\cite{Thanassoulis04, ManelliV06, HartR2015a}, leading to many papers studying simple mechanisms for optimizing revenue in different settings. For example, there is a stream of work providing simple mechanisms for set function valuations (such as additive, unit-demand, submodular) \cite{hart2017approximate, LiY13, BabaioffILW20, Yao15, RubinsteinW15, CaiDW16, ChawlaM16, CaiZ17}. We refer the readers to \cite{CaiZ17} for the state of the art on this. These results do not apply to our setting as they need the buyers' valuations for disjoint subsets of items to be independent which is not true in our setting. Also, these mechanisms take the form of VCG with entry fee. VCG-like mechanisms can perform poorly in the presence of intermediaries, need to be tailored to the demand structure and make the intermediaries' problem of bidding on behalf of their buyer's harder.
For the correlated values setting, \cite{HartN13} show that finite menu-size mechanisms don't perform well under arbitrary correlations even with two items and one bidder.\footnote{ However, when the seller is a max-min agent a pricing mechanism per item turns out to be optimal  \cite{Carroll2017}.} In contrast, we show that with a specific type of correlation, simple mechanisms with a small menu size perform well.

Another related sequence of papers \cite{ChawlaHK07, ChawlaHMS10, ChawlaMS10, AlaeiFHHM12, CaiDW12b, KleinbergW12} provide sequential posted-price mechanisms with a good approximation to the revenue. As noted in the introduction, using Corollary~\ref{coro:1}, these mechanisms can be used to obtain constant competitive mechanisms for our setting (see \Cref{prop:n}). 
In fact, a special case of our results (namely, the homogeneous item case with valuations drawn from an MHR distribution) can be obtained by combining the results of \citet{ChawlaHK07, ChawlaHMS10, ChawlaMS10} and a  modification of the result of \citet{kleinberg2013ratio}, which considers the ratio of the optimal revenue to optimal welfare and bounds this for the case of $c$-bounded regular distributions. We note that it is not clear how to extend this approach to the case of $\lambda$-regular distributions or the case of heterogeneous items. For the case of heterogeneous items, existing posted price mechanisms do not translate well to the intermediary setting and the result of \citet{kleinberg2013ratio} does not apply to $\lambda$-regular distributions. 

 \section{Model} \label{sec:model}

The baseline model consists of a seller with $k$ items to sell. We first consider the case where the items are homogeneous and later extend our results to the case of non-homogeneous items with separable valuations (see Section~\ref{sec:ext}). 
There are $n$ unit-demand buyers. For each $\bu \in [n]$, buyer $\bu$ has a valuation $v_\bu$ for obtaining an item. We assume that the valuations $v_\bu$ are drawn independently and identically distributed (i.i.d.) according to a distribution $F$ with support on  $\R_{+}$.

\bigskip 

A buyer may choose to bid directly in the auction or through an intermediary. Our model is agnostic to how these decisions (and negotiations) are made and, hence, we allow for an arbitrary partition of buyers across intermediaries ultimately bidding in the auction.\footnote{A buyer bidding directly in the platform corresponds to an intermediary exclusively representing him.} 

\begin{definition}[Demand Structure]
The {\em demand structure} $(m, \pi)$ of a market with $n$ buyers is composed of $m \leq n$ intermediaries and a partition $\pi(\cdot)$ of the set of buyers $[n]$. For an intermediary $\ag \in [m]$ $\pi(\ag)$ denotes the set of buyers that $\ag$ represents. 
\end{definition}

In particular, notice that our model includes the {\em competition case}, when all buyers are bidding directly in the auction ($m^C=n$, $\pi^C(\ag)=\{\ag\}$), and the {\em monopsony case}, when all buyers are represented by the same intermediary ($m^M=1$, $\pi^M(1)=[n]$). 

Although we consider the demand structure as exogenously given, our results can be easily extended to a setting where buyers strategically decide whether to participate directly in the auction or to be represented by a particular intermediaries (see Remark~\ref{rem:1} below for details).

\bigskip
\noindent{\bf Intermediaries' utility functions} \\

Motivated by two of the most natural and commonly-studied intermediary utility functions -- (a) surplus-maximizing intermediary which maximizes the surplus of the buyers it represents and (b) the monopolist intermediary that maximizes its profit from captive buyers -- we consider a general class of intermediary utility functions. Let $\vv_\ag= (v_\bu)_{\bu \in \pi(\ag)}$ be the valuations of the buyers represented by Intermediary $\ag$. Given $x_\ag \in [k]$ items allocated to Intermediary $\ag$ for a price $p_\ag \in \R_+$, we denote by $U_\ag(x_\ag,p_\ag;\vv_\ag)$ Intermediary $\ag$'s utility function. We apply the following natural restrictions: 
\begin{enumerate}
    \item (Individual Rationality) The Intermediary $\ag$ does not obtain any item, it gets a payoff of zero (i.e.~$\ag$ always opt out from the auction). 
    \item (Does not create value) For every vector of valuations $\vv_\ag$, allocation $x_\ag$, and price $p_\ag$ we have that
    \begin{align*} U_\ag (x_\ag,p_\ag;\vv_{\ag})\le V(x_\ag,p_\ag;\vv_{\ag}) :=  &\max_{z_\ag \in [0,1]^{\pi(\ag)}} \left\{ \sum_{\bu\in \pi(\ag)} v_{\bu}  z_{\bu \ag} -p_\ag\; \bigg| \; \sum_{\bu\in \pi(\ag)} z_{\bu \ag}\le x_\ag \right\}
    \end{align*}
    \item (\discon) When Intermediary $\ag$ can obtain the items for a uniform price of $r$ per item, Intermediary $\ag$ wants to purchase at least $\tau_\ag \cdot |\{ \bu: v_\bu \ge r\; \mbox{ for } \bu \in \pi(\ag)\}|$ items in expectation, for $\tau_\ag\in (0,1]$. In particular, for a surplus-maximizing intermediary $\tau_\ag = 1$ while for a monopolist intermediary $\tau_l\geq \frac 1 e$ whenever $F$ is MHR (see Lemma~\ref{lem:2} for details). 
    
\end{enumerate}
Condition 1 implies that intermediaries cannot obtain a negative payoff. Condition 2 captures that intermediaries do not generate value, hence, their payoff is bounded by $V(x_\ag,p_\ag;\vv_{\ag})$, the highest surplus Intermediary $\ag$ can generate among buyers in $\pi(\ag)$.  Condition 3 is a {\em \discon} condition: when items are not bundled and have a uniform per-item price $r$, if $s$ buyers in $\pi(\ag)$ would like to buy an item at the price $r$, then, independent of the valuations $\vv_{\pi(\ag)\setminus s}$, Intermediary $\ag$ wants to buy at least $\tau_\ag\cdot s$ items at that price. The probability $\tau_\ag$ captures the level of double-marginalization that may arise due to the downstream interaction between Intermediary $\ag$ and Buyers $\bu \in \pi(\ag)$. For example, when Intermediary $\ag$ is a monopolist (with respect to the buyers in $\pi(\ag)$), trying to maximize its own revenue, we have that $\tau_\ag \ge \frac 1 e$ if the distribution $F$ is MHR (see Lemma~\ref{lem:dm} for a precise statement).



The restrictions on the intermediaries' utility functions include models previously studied. Consider the following examples:

\begin{itemize}
\item A negotiation between buyers and intermediaries (with no asymmetric information) can lead to a utility function of the form
\begin{equation}\label{eq:1} U_\ag(x_\ag,p_\ag;\vv_\ag) = \alpha_\ag(v_\ag,x_\ag,p_\ag) V(x_\ag,p_\ag;\vv_\ag),\end{equation} 
where $\alpha_\ag(\vv_\ag,x_\ag,p_\ag)$ may represent either Intermediary $\ag$'s bargaining power or an auditing fee. In particular, when $\alpha_\ag\equiv 1$, Intermediary $\ag$ has full bargaining power. 
A simple exercise shows that for this class of utilities functions, the \discon condition holds with $\tau_\ag = 1$.  

\item
Intermediaries are profit-maximizers and lack knowledge of buyers' valuations. The intermediary can run a revenue maximizing truthful auction to maximize its profit.
Lemma~\ref{lem:dm} shows that these utilities functions satisfy our requirements with $\tau_\ag  = c(\lambda)$. 
\end{itemize}


{\bf The seller's problem.} The seller wants to implement a mechanism that maximizes her revenue. From the revelation principle \citep{myerson1981optimal}, we can restrict our attention to direct bayesian incentive compatible mechanisms (BIC). That is, $(x_\ag,p_\ag)_{\ag \in [m]}$, where $x_\ag:[\underline v, \overline v]^{n} \to \Delta([k])$,\footnote{For a set $X$, we denote by $\Delta (X)$ the set of probability measures on $X$.} $p_\ag:[\underline v, \overline v]^{n} \to \mathbb{R}_{+}$, satisfying  
\begin{align*}
    \E[U_\ag(x_\ag(\vv_\ag,\vv_{-\ag}),p_\ag(\vv_\ag,\vv_{-\ag});\vv_{\ag}) \; | \vv_{\ag}] &\ge \E[U_\ag(x_\ag(\vv'_\ag,\vv_{-\ag}),p_\ag(\vv'_\ag,\vv_{-\ag});\vv_{\ag}) \; | \vv_{\ag}]  & \mbox{(IC)} \\
     \E[U_\ag(x_\ag(\vv_\ag,\vv_{-\ag}),p_\ag(\vv_\ag,\vv_{-\ag});\vv_{\ag}) \; | \vv_{\ag}]  &\ge 0  & \mbox{(IR)} \\
 \sum_{\ag\in [m]} x_{\ag}(\vv) &\le k \qquad \mbox{ (a.s.)}.
\end{align*}
We denote by $\mbox{Rev}(F,m,\pi;(\vx,\vp)) = \E\bigg[\sum_{\ag\in [m]} p_\ag(\vv) \bigg]$ the revenue the seller receives by implementing mechanism $(\vx,\vp)$.

We define $\mbox{Rev}(F,m,\pi)$ to be the optimal expected revenue.

\begin{definition}[Optimal Expected Revenue, $\mbox{Rev}(F,m,\pi)$]\label{def:optrev}
$$\mbox{Rev}(F,m,\pi)= \max_{(\vx,\vp) \in \mbox{(BIC)}} \mbox{Rev}(F,m,\pi;(\vx,\vp)).$$
\end{definition}

\begin{definition}[\ostat{\ob}{t}]\label{def:ostat}
Define \ostat{\ob}{t} to be the $\ob^{th}$-order statistic of $\vv=(v_\bu)_{\bu\in [t]}$, where each $v_\bu$ is drawn i.i.d. according to the distribution $F$.
\end{definition}

\begin{definition}[Expected Optimal Welfare, $\mbox{Wel}(F)$]\label{def:optwel}
$$ \mbox{Wel}(F) =  \E\bigg[\max_{\vx\in \R^n} \sum_{\bu \in [n]} v_\bu x_\bu\bigg] = \E\bigg[\sum_{\ob=1}^k \ostat{\ob}{n}\bigg],$$
where \ostat{\ob}{n} is the $\ob^{th}$-order statistics of the random variable $\vv=(v_\bu)_{\bu \in[n]}$.
\end{definition}

\begin{remark}\label{rem:1}
For the purpose of a clearer exposition, our model restricts attention to cases where the demand structure is independent of the mechanism selected by the seller. However, our results can be easily generalized for the cases when $(m,\pi)$ is a function of $(\vx,\vp)$.
\end{remark}

\noindent{\bf Revenue benchmark:}
We will use Expected Optimal Welfare, $\mbox{Wel}(F)$ defined above as our revenue benchmark. Our use of this benchmark is motivated by two factors: (a) The maximum revenue a seller can extract can never exceed this benchmark (as noted in the lemma below), and (b) the optimal revenue mechanism (and the optimal revenue) is typically difficult to characterize in multi-parameter settings. 


\begin{restatable}{lemma}{lemzero}
\label{lem:0}
For every distribution $F$ and demand structure $(m,\pi)$, $\mbox{Rev}(F,m,\pi)\le \mbox{Wel}(F)$.
\end{restatable}

This result is standard in the auction literature and it naturally extends to our setting since intermediaries don't create value (Condition 2 on $U_\ag$). 

\bigskip
\subsubsection*{Distributional Assumptions}
We use $\phi(v):= v - \frac{1-F(v)}{f(v)}$ to denote the virtual valuation for a value $v$, drawn from a distribution with CDF $F$ and PDF $f$.  Throughout the paper we make the following assumption. 
\
\begin{definition}[$\lambda$-regularity]\label{def:lambda-regular}
Given $\lambda\in [0,1]$, a distribution $F$ is $\lambda$-regular if $\lambda v + \frac {1-F(v)}{f(v)}$ is non-decreasing on $v\in [\underline v, \overline v]$.
\end{definition}

\begin{assumption}[$F$ is $\lambda$-regular]\label{ass:1}
The distribution $F$ has a density $f>0$ on the support $[\underline v,\overline v]$ ($\underline v < \overline v \le \infty$), and is $\lambda$-regular for $\lambda <1$. 
\end{assumption}

Observe that for $\lambda = 0$, the condition is equivalent to the monotone hazard rate condition commonly used in economics literature, while for $\lambda=1$, the condition is equivalent to the Myerson regularity condition.  A distribution satisfying the $\lambda$-regularity has three important properties which are key to the robustness of our mechanism.

\begin{restatable}{fact}{factone}  \label{lem:1} 
Let $F^{(1)}$ the distribution of the first order statistic of $n$ i.i.d. random variables drawn according to $\lambda$-regular distribution $F$. Then, $F^{(1)}$ is $\lambda$-regular.
\end{restatable}

To the best of our knowledge, \Cref{lem:1} is not known prior to our work and we provide a full proof in the appendix. We use a nontrivial characterization of $\lambda$-regular distributions given by~\cite{schweizer2019performance}] to prove this fact and it might be of independent interest. 

We define the following notation which will be used throughout the paper: 
\begin{definition}($c(\lambda)$)
We define the quantity $c(\lambda) = (1-\lambda)^{\frac{1}{\lambda}}$ for all $\lambda \in (0,1]$ and define $c(0) = \lim_{\lambda \to 0} c(\lambda) = \frac{1}{e}$. 
\end{definition}

\begin{fact}\label{lem:2}[\cite{schweizer2019performance}] For a random variable $v$ drawn from a $\lambda$- regular distribution $F$, we have that $\P[v\ge \E[v]]\ge c(\lambda)$.
\end{fact}

\begin{fact}\label{cor:oneless}$\P[v^{(1,s-1)} \ge \E[v^{(1,s)}]] \geq \frac {s-1}s  c(\lambda)$.
\end{fact}

\begin{proof}
Draw $X_1, X_2, \ldots, X_s$ i.i.d. from F. Let \ostat{1}{s} be the maximum of $X_1, X_2, \ldots, X_s$, and let \ostat{1}{s-1} be the maximum of $X_1, X_2, \ldots, X_{s-1}$. From ~\Cref{lem:1}~and~\Cref{lem:2}, we know that
$ \P[ \ostat{1}{s} \geq \E[ \ostat{1}{s}]] \geq c(\lambda) $.
Since each of the $s$ draws is equally likely to be \ostat{1}{s}, the probability that \ostat{1}{s} lies in the first $s-1$ draws is $\frac{s-1}{s}$, giving
\begin{align*}
\P[ \ostat{1}{s-1} \geq \E[ \ostat{1}{s}]] &\geq \P[\ostat{1}{s} \in \{ X_1, X_2, \ldots, X_{s-1} \} ] \cdot \P[ \ostat{1}{s} \geq \E[ \ostat{1}{s}]] 
\geq \frac{s-1}{s } c(\lambda)
\end{align*}

\end{proof}

%
%

%
%

 \section{The importance of being robust}
\label{sec:robustness}
In this section, we argue the importance of implementing a mechanism that is robust to the details of the intermediaries' utility functions and to the different demand structures that may arise.

First, consider the dependence on intermediaries' utility functions. An intermediary's utility function arises from the downstream negotiations/contracts/mechanisms involving the individual buyers it represents. These downstream interactions may take many different forms and can vary by intermediary and change over time. It is hard for a seller to know about the details of these downstream interactions between each intermediary-buyer pair. And even in cases where the seller knows these details, the seller could be faced with a difficult mechanism design problem involving bidders with complex and heterogeneous multi-dimensional utility functions (see the literature review section in the introduction for a flavor of the difficulty of the problem). A mechanism that does not depend on the details of the intermediaries' payoff structures lets the seller sidestep these issues, additionally guaranteeing that it will continue to perform well even when the downstream contracts change.

Next, we argue the importance of implementing a mechanism that is robust to different demand structures. 
To this effect, we will show that natural revenue-maximizing mechanisms tailored for a particular demand structure can have perform quite poorly when faced with a different demand structure. If a seller implements a mechanism that depends on the demand structure, the seller should either know the demand structure beforehand with some certainty or construct a mechanism in which the seller learns the demand structure as part of the mechanism. We note that it is possible to construct a mechanism that learns the demand structure without incurring any cost (in one of its equilibria)\footnote{The first stage of the mechanism consists of asking each intermediary how many buyers it is representing, and if the total sum is not $n$, then the seller does not move to the second stage of the mechanism. Clearly, an equilibrium exists where intermediaries truthfully report the number of buyers they represent.} 
Even though such a mechanism may perform well in theory, the first stage of the mechanism is quite unstable --  there exist equilibria such that the mechanism never reaches the second step, resulting in zero revenue for the seller. 

In the rest of the section, we show that natural revenue-maximizing mechanisms designed for a particular demand structure $(m,\pi)$ can perform poorly when presented with a different demand structure $(m',\pi')$.
More precisely, we compare two extreme demand structures: competition $(m^C = [n],\pi^C(\ag)=\{\ag\})$ and monopsony $(m^M=1,\pi^M(1)=[n])$. For this exposition, we assume that intermediaries have complete information and full bargaining power, i.e., their utilities functions are $U_\ag(x_\ag,p_\ag;\vv_\ag)= V(x_\ag,p_\ag;\vv_\ag) = \sum_{\ob=1}^{x_\ag} \ostat{\ob}{\pi(\ag)} - p_\ag$.   

Along with showing that robustness is critical in situations where the demand structure is unknown, we observe a surprising result that is of independent interest. We show that there are instances where a seller prefers facing a monopsony demand structure rather than full competition. Even though competition is helpful in raising prices in an auction, facing one bidder that demands multiple items allows the seller to construct bundle mechanisms which are known to be more effective compared to simply selling the items separately \citep{mcafee1989multiproduct}. In particular, we observe that the competition effect is particularly helpful in situations when there is scarcity of items (small $k$) while the bundling effect is helpful in the presence of many items (high $k$).


\subsection{Optimal mechanism for the competition case}
We first characterize the optimal auction for the competition case and show that in some instances, it performs poorly for the monopsony case.

\begin{restatable}{proposition}{propone} 
\label{prop:1}
In the competition case, the revenue-maximizing mechanism can be implemented with $(k+1)^{th}$-price auction with a reserve price $r =\phi^{-1}_{F}(0)$.  
\end{restatable}
Although this statement is a well-known result in the literature, for the sake of completeness we provide a complete proof in the appendix. 

To simplify notation, we define $\mbox{Rev}^C(\vx,\vp)=\mbox{Rev}^C(F,m^C,\pi^C;(\vx,\vp))$ and $\mbox{Rev}^C = \mbox{ Rev}^C(F,m^C,\pi^C)$. Likewise, we define Rev$^M(\vx,\vp)$ and Rev$^M$ for the monopsony case. 

The next proposition shows that there exists an instance where the revenue obtained by the optimal mechanism for the competition case is not constant-competitive when the demand structure is monopsony (we defer the proof to the appendix).

\begin{restatable} {proposition}{proptwo}
Let $(\vx^C,\vp^C)$ be the optimal mechanism for the competition case. Then for every $\gamma>0$, there exist $n, F, k$ such that Rev$^M(\vx^C,\vp^C)<\gamma$Rev$^M$.
\end{restatable}

%

\subsection{Optimal mechanism for the monopsony case}

The previous result shows that implementing the optimal mechanism for the competition case can have a low performance when the demand structure is a monopsony. To show the converse, that the optimal mechanism for the monopsony case has a low performance on the competition setting, is technically challenging since we need to solve a multi-item mechanism design problem. Even for $k=n$, when the intermediary's valuations for the items are i.i.d, we have that there are instances where the optimal mechanism consists of a complex set of menus \citep{daskalakis2017strong}. For the case $k<n$, when valuations are correlated, finding the optimal mechanism is believed to be hard, and to the best of our knowledge, it remains an open problem. That being said, the next result shows that the simple mechanism of selling all items as a bundle gets a constant factor of the optimal revenue, and when $k=n$, it is asymptotically optimal. This result suggests that for $k=n$ a seller facing a monopsony should implement a (near-optimal for large $n$) pure bundling strategy. Once again, due to space constraints, we will defer the proofs to the appendix. 

\begin{restatable}{proposition}{propmonopsony} \label{prop:2}
Suppose that $F$ is MHR, then for the monopsony case consider the mechanism selling all $k$ items in a bundle for a price $p_B=\phi_{G}^{-1}(0)$ where $G$ is the distribution of $\sum_{i=1}^k \ostat{i}{n}$. Then, the above mechanism provides at least $\frac 1 e$ of the optimal mechanism. Moreover for any distribution $F$ with finite moments we have that when $k=n$, the revenue of selling the bundle at a price $p^{\epsilon}_B=  n(\E[v]-\epsilon)$ is asymptotically optimal as $n\to \infty$ and $\epsilon \to 0$.\footnote{\citet{csorgo1989limit, broniatowski1998strong} show an asymptotic result for the case $k<n$ when the hazard rate $h_F$ has a linear growth and $k(n)$ does not increase too fast with $n$.}
\end{restatable}

%
%

Using the previous result we obtain the surprising result that a monopsony may be preferable for the seller over competition.
\begin{restatable} {corollary}{corgoodmonopsony}
There are some instances where R$^M > $R$^C$.
\end{restatable}

%
Next we show that the bundle pricing mechanisms can make very little revenue when faced with Competition demand structure, compared to the optimal mechanism for Competition.

\begin{restatable}{proposition}{propbundlebad}
Consider the case $k=n$ and suppose that the seller implements the asymptotically optimal bundle pricing mechanism $(\vx^B,p_B^\epsilon)$. Then for every $\gamma>0$, there exists an instance $F$ such that Rev$^C(\vx^B,p_B^\epsilon)<\gamma$Rev$^C$.
\end{restatable}

\section{The \robustmechanism}\label{sec:robustmechanism}

This section constructs the \robustmechanism (IPM) and shows that for any demand structure $(m,\pi)$ and any intermediaries' utility functions $(U_\ag)_{\ag \in [m]}$ satisfying conditions (1)-(3), the \robustmechanismshort obtains a constant factor of the welfare and, therefore, of the optimal expected revenue $\mbox{Rev}(F,m,\pi)$. As a corollary, we obtain that the revenue-maximizing mechanism for each particular market structure is always a constant factor of the welfare.   


We begin by identifying a key property of order statistics of i.i.d.~ samples from a distribution $F$. We show that expectation of the highest value among $\ceil{n/k}$ i.i.d.~samples is at least a constant fraction of the average of the expectations of the top $k$ samples among $n$ i.i.d.~samples.\footnote{Because intermediaries' don't create value (Condition 2) and buyers are unit-demand, without loss of generality we can assume that $k\le n$.} We note that the above fact is true for all  distributions. The lemma below will serve as the basis for the \robustmechanism. 

\begin{lemma}\label{lem:main}
Recall that $\ostat{\ob}{t}$ is  the $\ob^{th}$-order statistics of $t$ i.i.d.~random variables drawn according to $F$. If $n \geq k$, 
\[ k \cdot \E[\ostat{1}{\ceil{n/k}}] \ge \left(1-\frac 1 e \right) \left (\sum_{\ob=1}^k \E[\ostat{\ob}{n}] \right), \]
\end{lemma}
\begin{proof}
Let $X = \set{ v_{1},\dots v_{n}}$ be a set of $n$ i.i.d. random variables drawn from $F$. 
We denote the $k$  largest random variables to be  $\ostat{1}{n}, \ostat{2}{n}, \dots, \ostat{k}{n}$  
and we shall refer to them as ``large'' elements. 
Let $s = \ceil{n/k}$. One way to choose $s$ random variables uniformly at random is to choose a subset of size $s$ uniformly at random from $X$. Let $Y = \set{ Y_1, \dots Y_{s}}$ be a uniformly random subset of X of size $s$.
If the set $Y$ contains a large element, then the $\max_\bu{ Y_{\bu}}$ must be a large element. 
In expectation, the value of a large element is $\sum_{\ob =1}^k \E[\ostat{\ob}{n}]/ k$. 
We will now show that the probability that $Y$ contains no large elements is at most $\frac1e$.

To see this, observe that
\begin{align*}
  \P[Y \text{ contains no large elements}] &= \frac{ \binom{n-k}{s} }{ \binom{n}{s} } \\
   &=  \frac{ (n-k)! }{n! } \cdot \frac{(n-s)!}{(n-k-s)!}  \\
    &= \prod_{\ob=1}^k \left(1- \frac{s}{n-k+\ob} \right) \\
    &\leq  \prod_{\ob=1}^k \left(1- \frac{s}{sk-k+\ob} \right) 
    \leq (1- 1/k)^k \leq 1/e
\end{align*}
  
Thus we get that $\E[\ostat{1}{s}]  \geq  (1-\frac1e) \cdot \frac{1}{k}\E[\sum_{\ob=1}^k \ostat{\ob}{n}]$.  
%
\end{proof}


We are now in position to present the \robustmechanism. 
\begin{definition}[Intermediary-Proof  Mechanism]\label{def:ipm_hom}
 Let $s=\lceil n/k \rceil$. The  \robustmechanism $(\vx^R,\vp^R)$ consists of a simple uniform single-item pricing scheme where each item is sold separately at a price $p^R=\E[v^{(1,s)}]$. In the case of excess demand, items are randomly allocated among the interested bidders.
\end{definition}

Remarkably, the \robustmechanismshort is a simple single-item posted-price that depends only on the number of items $k$, number of buyers $n$ and the distribution $F$. Thus, the seller does not have to learn the demand structure nor the intermediaries' payoff. In other words, the \robustmechanismshort is completely robust to the intermediaries' involvement in the auction. Besides the simplicity of its implementation, posted prices have appealing properties for the bidders' behavior \citep{goldberg2005}.

We next prove that the \robustmechanism obtains a constant factor of the optimal revenue. The proof proceeds by dividing the $n$ bidders into $k$ groups of roughly $\lceil n/k \rceil$ bidders and estimating the revenue from selling a single item to each group. 


\begin{theorem}\label{th:1} 

Let $\tau = \min_{\ag \in [m]} \tau_\ag$, then for every distribution $F$ and demand structure $(m,\pi)$ we have that
$$ \mbox{Rev}(F,m,\pi;(\vx^R,\vp^R) ) \ge  \tau \cdot \frac{c(\lambda)}{2} \cdot \left(1-\frac 1 e\right)  \mbox{Wel}(F) \ge \tau \cdot \frac{ c(\lambda) }{2 } \cdot \left(1-\frac 1 e\right)  \mbox{Rev}(F,m,\pi).$$
Recall $c(\lambda) = (1-\lambda)^{\frac{1}{\lambda}}$ for $\lambda \in (0,1]$ and $c(0) =  \lim_{\lambda \to 0} c(\lambda) = \frac{1}{e}$. Also recall that $\lambda = 0$ for MHR distributions.
\end{theorem}
\begin{proof}
We start by dividing the $n$ valuations $v_\bu$, $\bu \in [n]$, randomly into $k$ sets of size $s'$, where $s' = 1$ if $s=1$, and $s' = s-1$ otherwise (there might be some leftover valuations after we make the sets). We note that the maximum valuation in each of these sets is distributed according to \ostat{1}{s'}.

The \discon ~~condition on $U_\ag$ (Condition 3) implies that Intermediary $\ag$ will try to buy at least $\tau \cdot |\{\bu : v_\bu \ge p^R, \bu \in \pi(\ag)\}|$ items in expectation. 
Using this, we lower bound the revenue as follows.
\begin{align*}
 \mbox{Rev}(F,m,\pi;(\vx^R,\vp^R) ) &\ge \tau \sum_{\bu = 1}^n  \P[\text{Buyer }\bu \mbox{ has value at least } p^R ] \cdot p^R\\
 &\geq \tau \cdot k\cdot \P[\ostat{1}{s'}\ge p^R] \cdot p^R \\
 &\ge \tau \cdot k \cdot \frac{s'}{s} \cdot c(\lambda)   \cdot \E[\ostat{1}{s}]\\
 &\ge \tau \cdot \frac{s'}{s} \cdot c(\lambda)  \cdot \left(1-\frac 1 e\right) \cdot \underbrace{\left (\sum_{\ob=1}^k \E[\ostat{\ob}{n}] \right)}_{\mbox{Wel}(F)}\\
 &\ge \tau \cdot \frac{c(\lambda) }{2} \cdot \left(1-\frac 1 e\right) \mbox{Wel}(F)\\
 &\ge \tau \cdot \frac{c(\lambda)}{2}  \cdot \left(1-\frac 1 e\right)  \mbox{Rev}(F,m,\pi).
\end{align*}
The third inequality holds because of \Cref{cor:oneless} and the definition of $p^R$. The fourth inequality comes from Lemma~\ref{lem:main}. The last inequality is due to Lemma~\ref{lem:0}.
\end{proof}

In the proof of the theorem we observe that $1/2$ factor is not necessary when $s'=s$, i.e., when $n/k$ is an integer number. The following corollary remarks this observation.
\begin{corollary}
Suppose that $n/k$ is an integer number. Then,
$$ \mbox{Rev}(F,m,\pi;(\vx^R,\vp^R) ) \ge  \tau c(\lambda) \left(1-\frac 1 e\right)  \mbox{Wel}(F) \ge \tau c(\lambda) \left(1-\frac 1 e\right)  \mbox{Rev}(F,m,\pi).$$
\end{corollary}

A second consequence of Theorem~\ref{th:1}, which is of independent interest, shows that for every demand structure and bidders' utility functions satisfying Conditions (1) - (3), the optimal revenue is a constant factor of the optimal welfare.

\begin{corollary}\label{coro:1}
For every distribution $F$ and demand structure $(m,\pi)$, we have that $$\mbox{Rev}(F,m,\pi) \ge \tau  \frac{c(\lambda)}{2}\left(1-\frac 1 e\right)\mbox{Wel}(F).$$
\end{corollary}


The above corollary enables us to show that per-item posted-price mechanisms that perform well in the full competition setting (when all buyers participate directly in the auction) perform well in the context of intermediation as well when conditions (1)-(3) are satisfied. This is because per-item posted-price mechanisms provide, up to a constant factor $\tau$, the same revenue for the seller under all demand structures. Combined with Corollary~\ref{coro:1}, this implies that if a per-item posted-price mechanism is a constant approximation of the optimal revenue for a particular demand structure (e.g. full competition), then it is a constant approximation of the welfare, and hence the optimal revenue for any demand structure.\footnote{\citet{ChawlaHK07} construct a uniform per-item posted-pricing mechanism that is at least a 14\% approximation of the optimal mechanism when buyers are directly bidding in the auction (competition case). Corollary~\ref{coro:1} guarantees that this mechanism provides a 
$\tau \;  \frac{0.14}{2e}\left(1-\frac 1 e\right)$ approximation of the welfare and also the revenue for any demand structure.} We now extend this to a more general class of mechanisms.

\begin{definition}[$\beta$-Robust Mechanism]
A mechanism $(\vx,\vp)$ is a $\beta$-Robust Mechanism if for every distribution $F$ and demand structures $(m,\pi)$, $(m',\pi')$ we have that $\mbox{Rev}(F,m',\pi';(\vx,\vp) )\ge \beta \; \mbox{Rev}(F,m,\pi;(\vx,\vp) )$.\footnote{Note that from the \discon condition on the 
intermediaries' utility function, every per-item posted-pricing mechanism is a $\tau$-Robust Mechanism.}
\end{definition}

\begin{proposition}\label{prop:n}
Consider a $\beta$-Robust Mechanism $(\vx,\vp)$ such that the revenue for the demand structure $(\overline m,\overline \pi)$ is a $\gamma$ approximation of $\mbox{Rev}(F,\overline m,\overline \pi)$. Then for every distribution $F$ and demand structure $(m,\pi)$ we have that
\begin{align*}\mbox{Rev}(F,m,\pi;(\vx,\vp) ) \ge  \beta \gamma \tau \frac{c(\lambda)}{2 } \left(1-\frac 1 e\right)  \mbox{Wel}(F) \ge  \beta \gamma \tau \frac{ c(\lambda)}{2} \left(1-\frac 1 e\right)  \mbox{Rev}(F,m,\pi). 
\end{align*}
\end{proposition}
\begin{proof}
We do the proof in two steps. First, because $(\vx,\vp)$ is  a $\beta$-Robust Mechanism we have that $\mbox{Rev}(F,m,\pi;(\vx,\vp) )\ge \beta \; \mbox{Rev}(F,\overline m,\overline \pi;(\vx,\vp) )$. Second, Corollary~\ref{coro:1} implies that $\mbox{Rev}(F,\overline m,\overline \pi;(\vx,\vp) ) \ge \gamma \tau  \frac{c(\lambda)}{2}\left(1-\frac 1 e\right)\mbox{Wel}(F)$. We conclude using that $\mbox{Wel}(F) \ge \mbox{Rev}(F,m,\pi)$. (Lemma~\ref{lem:0}).
\end{proof}

To finish this section we present two important applications of Theorem~\ref{th:1}. The first one is when the intermediaries are aware of the buyers' valuations, and their payoff is a fraction of the overall surplus (see Equation~\eqref{eq:1}). In that case, $\tau_\ag= 1$ for every $j$.

\begin{corollary}
When intermediaries utilities' functions are of the form described in Equation~\eqref{eq:1} the robust mechanism gets at least a $\frac{c(\lambda)}{2}  \left(1-\frac 1 e\right)$ fraction of the Optimal welfare (and revenue). 
\end{corollary}

The second application is the case of double marginalization when each intermediary does not know the buyer's valuation and decides to implement the profit-maximizing mechanism. The following lemma shows that for that case $\tau_\ag \ge \frac 1 e$.

For the sake of exposition, we state and prove the following lemma first for the case when the distribution is MHR(i.e.~ $\lambda=0$). 
We present the general lemma below and defer the proof to the appendix. 

\begin{lemma}\label{lem:dm}
Suppose that Intermediary $\ag$ is a monopolist and only knows the prior of the buyer's valuations. If the seller posts a price $p$ per item and Buyer $\bu\in \pi(\ag)$ has valuation $v_\bu \ge p$ then the probability that Intermediary $\ag$ would like to buy one item is at least $\frac 1 e$.
\end{lemma}

\begin{proof}
    Because the price per item is $p$, the mechanism design problem for Intermediary $\ag$ is the same one solved in the Proof of Proposition~\ref{prop:1} when the cost of production is $p$. 
    Thus, Intermediary $\ag$ would like to allocate the item to Buyer $\bu$ if $\phi_F(v_\bu)\ge p$. Thus, we need to show that $\P[\phi_F(v_\bu)\ge p | v_\bu \ge p] \ge \frac 1 e$.
    
    From Bayes' rule we derive that 
    \begin{align*}\P[\phi_F(v_\bu)\ge p | v_\bu \ge p] &= \frac{\P[\phi_F(v_\bu)\ge p , v_\bu \ge p]}{\P[ v_\bu \ge p]} =  \frac{\P[\phi_F(v_\bu)\ge p]}{\P[ v_\bu \ge p]}\\&=\frac{1-F(\phi^{-1}_F(p))}{1-F(p)}.
    \end{align*}
    The second equality follows from the fact that $\phi_F(v_\bu)  = v_\bu - (1 - F(v_\bu))/f(v_\bu) \le v_\bu$.
    Also, note that $1-F(v)= e^{-H(v)}$ where  $H(v) = \int_0^v h_F(z) dz$. Then, we have that
    \begin{align*}\P[\phi_F(v_\bu)\ge p | v_\bu \ge p] &=e^{-(H(\phi^{-1}_F(p)) - H(p))}\\
    &\ge e^{-\max_{p \ge 0}(H(\phi^{-1}_F(p)) - H(p))}\\ 
    &= e^{-\max_{p\ge \phi^{-1}(0)}(H(p) - H(\phi_F(p)))},
    \end{align*}
    where the last equality holds because $\phi_F$ is increasing. 
    
    To conclude the proof, we assert that $H(p) - H(\phi_F(p))\le 1$. Indeed, because $F$ is MHR we have that $H$ is convex with $H'(v) = h_F(v)$. The convexity and differentiability of $H$ implies that 
    \begin{equation}\label{eq:convex} H(\phi_F(p)) \ge H(p) + h_F(p) (\phi_F(p) - p ).\end{equation}
    From the definition of the virtual valuation we have that $\phi_F(p) = p - \frac 1 {h_F(p)}$. Replacing in the previous inequality we conclude that   $H(p) - H(\phi_F(p))\le 1$. 
\end{proof}

\begin{restatable}{lemma}{lemdmone} \label{lem:dml}
Suppose that Intermediary $\ag$ is a monopolist and only knows the prior of the buyer's valuations. If the seller posts a price $p$ per item and Buyer $\bu\in \pi(\ag)$ has valuation $v_\bu \ge p$ then the probability that Intermediary $\ag$ would like to buy one item is at least $ c(\lambda)$. Recall $c(\lambda) = (1-\lambda)^{\frac{1}{\lambda}}$ for $\lambda \in (0,1]$ and $c(0) =  \lim_{\lambda \to 0} c(\lambda) = \frac{1}{e}$. 
\end{restatable}
We defer the proof to the appendix.

\begin{corollary}
Suppose that each intermediary is behaving as monopolist and implementing the profit-maximizer mechanism. Then, the robust mechanism has a $\frac{c(\lambda)^2}{2} \left(1-\frac 1 e\right)$ revenue guarantee. 
\end{corollary}

\section{Selling non-homogeneous items}\label{sec:ext}


In this section, we consider the problem when the items are not necessarily homogeneous. We consider $k$ items with weights $\eta_1 \geq \eta_2 \geq \ldots \geq \eta_k \geq 0$ so that the value for buyer $i$ of getting item $\ob$ is $\eta_\ob v_\bu$. Given that the items are non-identical, we will adapt the mechanisms the seller can implement to incorporate per-item allocation. In this setting, a BIC mechanism consists of $(x_{\ob\ag},p_\ag)_{\ob\in[k],\ag\in[m]}$, where $x_{\ob\ag}$ is the probability that Intermediary $\ag$ get item $\ob \in [k]$ and $p_\ag$ corresponds to the price Intermediary $\ag$ has to pay.

In this generalized model, the expected welfare is given by $\mbox{Wel}(F) = \sum_{\ob=1}^k \eta_\ob \E[\ostat{\ob}{n}]$.

Observe that the mechanism for the identical items case will not perform well in the general setting. \robustmechanism does not distinguish between the case of uniform item weights and the case of one item with a high weight and the rest with near-zero weights. In the latter case, we need to charge close to $E[\ostat{1}{n}]$ in order to compete with welfare, and this could be much higher than 
$E[\ostat{1}{\ceil{n/k}}]$, e.g. for $k=n$.


{\bf Condition on utility functions of Intermediaries:} We observe that in order to guarantee good welfare, higher-weight items need to cost more on a per-unit-weight basis. Otherwise there will be no way to screen for higher value buyers to assign them the higher-weight items. In addition, we need to ensure that the utility functions are such that the intermediary will choose higher-value, higher-cost items over lower-value, lower cost items in some cases.
With this in mind, we will impose the following, relatively strong, condition on the utility functions of the Intermediaries for the purpose of this section. We will assume that the utility function of Intermediary $\ag$ is such that it tries to maximize the surplus among its buyers, i.e., 
 $U_\ag (x_\ag,p_\ag;\vv_{\ag})= V(x_\ag,p_\ag;\vv_{\ag}) :=  \max_{z_\ag \in [k]^{\pi(\ag)}} \left\{ \sum_{\bu\in \pi(\ag), j \in [k]} v_{\bu}  \eta_l z_{\bu j \ag} \bigg| \; \forall l \;\sum_{\bu \in \pi(\ag)} z_{\bu j \ag}\le x_{\ag,j} \right\}  - p_\ag.$
 Here $x_{\ag,j}$ is the fraction of item $j$ assigned to the intermediary and $z_{\bu j \ag}$ is the fraction of item $j$ that the aggregator assigns to $\bu$. 
Let $z^*_l(x_l)$ be the $z_l$ that maximizes the above expression. Then intermediary $\ag$ gets a utility of
$U_\ag(x_\ag,p_\ag;\vv_\ag) = \sum_{\ob=1}^{x_\ag} \sum_{j = 1}^k v_i \eta_j z^*_{ijl} - p_\ag$.


As we argued in the previous paragraph, it is important to ensure that items with higher weights are charged a higher price to ensure that we can screen for buyers with a high value. However, we need to make sure that the balance of prices is such that the high-value items are still bought with a reasonable probability.





\subsection{A robust mechanism for heterogeneous items}


In this section, we give a sequential posted-price mechanism with non-uniform per-item prices, and show that it gets a revenue that is a constant fraction of the optimal welfare.

\begin{definition}[Intermediary-Proof  Mechanism  for  Heterogeneous Items]\label{def:het}
Let $u_\ob = E[\ostat{1}{\ceil{n/\ob}}]$. Define ${r}_\ob =  {r}_{\ob+1}  + u_\ob \cdot (\eta_\ob - \eta_{\ob+1})$.
\heteromechanism ~~ is a sequential posted-price mechanism that offers the same menu of prices to each Intermediary, ordering the Intermediaries in an arbitrary order. The menu consists of individual prices for each item, with the price of item $\ob$ being $r_\ob$, and each Intermediary allowed to buy any items it likes as long as the item is still available.
\end{definition}
Note that $u_\ob$ and $r_\ob$ are decreasing in $j$. Thus the intermediary pays more for positions with higher $\eta_j$. 

We will show that due to the choice of $r_\ob$s, an intermediary representing a buyer with value $v_\bu$ will purchase item $\ob$ if $v_\bu \in [u_\ob, u_{\ob-1})$ (or another item with the same $\eta$ as $\ob$), since this item will generate the most value through buyer $\bu$.
Below we make this more formal for intermediaries representing multiple buyers.

\begin{definition}[Demand-set, $D(\ob)$]
For any buyer $\bu$, let $B(\bu)$ be the lowest-indexed item $\ob$ s.t. $v_\bu \geq u_\ob$.
Let demand-set of item $\ob$ be $D(\ob) = \{ \bu \in [n] | B(\bu) = \ob \} $.
Note that $D$ is a partition of the set of buyers.
\end{definition}

\begin{lemma} \label{lem:buyone}
When an intermediary $\ag$ is offered the above menu of prices, one optimal choice for the intermediary is to buy at least those available items $\ob$ for which the demand-set $D(\ob)$ includes a buyer from $\pi(\ag)$.
\end{lemma}
\begin{proof}
Recall that the intermediary is trying to select a subset of available items that maximizes 
$U_\ag(x_\ag,p_\ag;\vv_\ag) = \sum_{\ob=1}^{x_\ag} \sum_{j=1}^k v_i z^*_{ijl} - p_\ag$.
Contrary to the above claim, let the intermediary not buy an item $\ob$ even though it has a buyer $\bu \in D(\ob)$. Let $z^*_{\bu j \ag}$ denote its allocation. If $\sum_{j=1}^ k z^*_{\bu j \ag} < 1$, then the Intermediary can clearly increase its utility by buying $\ob$ fractionally and assigning it to $\bu$. If $\sum_{j=1}^ k z^*_{\bu j \ag} = 1$, then we will show that the Intermediary can buy item $\ob$ instead of the one currently assigned to $\bu$, without any loss in utility. Consider the case where buyer $\bu$ is assigned an item $\ob' < \ob$. Note that $v_\bu < u_{\ob''}$ for all $\ob'' < \ob$. By the construction of $r_{\ob}$'s, 
$$r_{\ob'} - r_{\ob}  = \sum_{k=\ob'}^{\ob-1} u_k (\eta_k - \eta_{k+1})\geq v_\bu (\eta_{\ob'} - \eta_\ob).$$ This inequality is strict when $\eta_\ob' > \eta_\ob$. 
Rearranging we get, 
$v_\bu \eta_\ob - r_{\ob} \geq v_\bu \eta_{\ob'} - r _{\ob'}$. The intermediary may not necessarily allocate the new item to $\bu$. However in allocating the item to somebody else the intermediary will further increase its utility.
Therefore, the Intermediary can swap item $\ob'$ for $\ob$ without any loss in utility. This process can be continued for multiple items until a whole unit $\ob$ is purchased.

Suppose instead that  buyer $\bu$ is assigned an item $\ob' > \ob$. For all $\ob'' \geq \ob$, $v_{\bu} > u_{\ob''}$. By the construction of $r_{\ob}$s, 
$$r_{\ob} - r_{\ob'}  = \sum_{k=\ob}^{\ob'-1} u_k (\eta_k - \eta_{k+1})\leq v_\bu (\eta_{\ob} - \eta_{\ob'}),$$ 
and we conclude that the intermediary would prefer to swap $\ob$ instead of $\ob'$
\end{proof}

Note that the Intermediary might buy more items than the ones specified in the above lemma.


\begin{restatable}{theorem}{thmheterorevenue} \label{thm:heterorevenue}
The \heteromechanism obtains a revenue of at least 
$$ \left(1- e^{-c(\lambda)/2} \right) \cdot (1 - \frac 1 e) \mbox{Wel}(F). $$
\end{restatable}

\begin{proof}

Let $v$ be a random draw from $F$, we denote by $p_\ob = \P[u_{\ob-1} \geq v \geq  u_\ob] $.

From~\Cref{lem:buyone}, which holds independently of the shape of the distribution $F$,  we know that if the valuation of at least one
buyer lies in the range $[u_{\ob-1},u_{\ob}]$, then item $\ob$ will be bought. 
Thus the probability that item $\ob$ is sold is
\begin{align*}
    \P[\text{item }\ob\; \text{is sold} ] &\geq 1- \P[\text{all values lie outside}[u_{\ob-1},u_\ob ] \\
    &= 1- (1-p_\ob)^n \\
    &\geq 1- e^{-p_\ob\cdot n}
\end{align*}

Therefore the revenue obtained will be at least 
\begin{align*}
   \min_{p_\ob} \sum_{\ob=1}^k r_\ob &\cdot   (1- e^{-p_\ob \cdot n})  
\end{align*}
Next, we note what we know about the $p_\ob$s. First, $0 \leq p_j \leq 1$. Further, \Cref{cl:1}, stated and proved in the appendix, shows that for $ \ob \geq 1$ we have that
$$n \cdot \sum_{t \leq j} p_t \geq  \frac{\ob}{2}(1-\lambda)^{\frac 1 \lambda}.$$    

The above system can be written as a convex program
\begin{align*}
   \max \sum_{\ob=1}^k r_\ob &\cdot    e^{-p_\ob \cdot n}  \label{prog:P1} \tag{Primal}\\
   n \cdot \left(\sum_{j=1}^s p_j \right) &\geq \frac{s}{2} (1-\lambda)^{\frac 1 \lambda} \qquad \text{ for all } 1\leq s \leq k \\
   0 \leq p_\ob &\leq 1  \qquad \qquad \qquad  \text{ for all } 1\leq \ob \leq k
\end{align*}

From~\Cref{lem:optprog}, stated and proved in the appendix, we get that the minimum
value of the above program is achieved when  $p_1 =\dots =p_k = \frac{1}{2n} (1-\lambda)^{\frac 1 \lambda}$. Plugging this into the objective, the revenue obtained is at least
\begin{align*}
    & (1 - e^{-1/2(1-\lambda)^{\frac 1 \lambda} }) \cdot \sum_{i=1}^k r_i \\
    &= (1 - e^{-1/2(1-\lambda)^{\frac 1 \lambda} }) \cdot \sum_{i=1}^k i \cdot (\eta_i - \eta_{i+1}) \cdot u_i \\
    &\geq (1 - e^{-1/2(1-\lambda)^{\frac 1 \lambda} })  \left(1-\frac{1}{e} \right) \cdot \sum_{i=1}^k  (\eta_i -\eta_{i+1}) \cdot \E[ \sum_{j=1}^{i} \ostat{j}{n} ] \\
\intertext{collapsing the sum and noting that $\eta_{k+1}=0$,}
    &= (1 - e^{-1/2(1-\lambda)^{\frac 1 \lambda} })  \left(1-\frac{1}{e} \right) 
    \cdot \sum_{i=1}^k \eta_i \cdot  \ostat{i}{n}
\end{align*}
The latter quantity is the welfare and this concludes the proof. 
\end{proof}

 \section{Open problems}\label{sec:discussion}

In this section, we describe some open problems in designing mechanisms that are robust to intermediaries.
\bigskip

\noindent {\bf Regular distributions:} Does an intermediary-proof mechanism exist for the case when the buyers' valuations are regular? Recall regular distributions are $\lambda$-regular with $\lambda=1$. This problem requires a different approach than the one used in our work: with regular distributions, using welfare as the revenue benchmark will not allow us to show a good approximation ratio.\footnote{For example, the welfare for the equal revenue distribution $F(v)= 1-\frac 1 v$ for $v\in [1,\infty)$, which is a regular distribution, is unbounded.}. Hence, the difficulty lies in finding an alternate characterization of the optimal revenue that, ideally, does not go through understanding the optimal mechanism for each demand structure, which is known to be a hard problem \citep{DaskalakisDT14}. Furthermore, the following lemma shows that a simple mechanism that posts a price per item cannot achieve a constant approximation ratio.
\begin{restatable}{lemma}{lemregular} \label{lem:regular-monopsony}
There exists a regular distribution $\mathcal{F}$ such that with $n$ items to sell and $n$ buyers represented by an intermediary with values drawn i.i.d.~from distribution $\mathcal{F}$, the optimal revenue from item pricing is $\Omega(\ln n)$ worse than the optimal revenue. 
\end{restatable}
The proof of Lemma~\ref{lem:regular-monopsony} considers a distribution closely related to the equal revenue distribution. On the one hand, we show that the optimal item pricing yields at most $O(n)$ in revenue irrespective of the demand structure. While, on the other hand, using concentration arguments we show that a bundle price of $\Theta(n \log n)$ will be accepted by a monopsony intermediary with a constant probability. We defer the formal proof to the appendix.\footnote{If we restrict the demand structures to be either a monopsony (a single intermediary representing all buyers) or the competition case (each buyer represents themselves), we can obtain a $2$-approximation mechanism by uniformly randomly choosing between item prices and a single bundle price.} \medskip
\bigskip

\noindent {\bf Non-identical buyer valuations, more general heterogeneous valuations:} Another natural problem is to design an intermediary-proof mechanism when buyers can have non-identical and/or general (non-separable) item-specific valuations for the items. We expect that the intermediary-proof mechanism for the non-identical buyers setting would require personalized pricing for each buyer and, when there are complementarities between the goods, pricing for bundles would also be needed. \medskip

\bigskip

\noindent {\bf Characterization of the optimal mechanism for the monopsony case:} Another interesting problem comes from the connection between an intermediary representing multiple buyers and a single buyer with valuations for multiple items. More precisely, a surplus-maximizing intermediary can be interpreted as a single buyer whose valuations for the $k$ items correspond to the $k$-highest valuations of the buyers the intermediary represents. Thus, our intermediary framework microfounds a natural type of correlation between the buyer's valuations for the multi-item setting. It would be interesting to characterize the optimal mechanism for this natural multi-item setting: single buyer whose valuations for the $k$ objects are drawn from the $k$-highest order statistics of $n$ i.i.d. draws. This question has been partially answered for two extreme cases: for $k=1$ the optimal mechanism is a simple posted price \citep{myerson1981optimal}, while the case of $k=n$ is known to be a hard problem \citep{DaskalakisDT14}. It would be interesting to study how the complexity of the optimal mechanism changes as a function of $k$.

\bibliographystyle{plainnat}
\bibliography{ref}

 \appendix
 \section{Proof of  \Cref{lem:1} }
For the remaining proofs, we will use the following key property of $\lambda$-regular distributions.
\begin{proposition}[Proposition 1, \cite{schweizer2019performance}]\label{prop:schw}
A distribution $F$ is $\lambda$-regular if and only if an increasing convex function $H_\lambda$ exists such that $1-F(v)= \Gamma_\lambda(H_\lambda(v))$ where $\Gamma_\lambda(v) = (1+\lambda v)^{-\frac 1 \lambda}$. Moreover, $H_\lambda(v) = \int_0^v r_\lambda(z) dz$ where $r_\lambda(v) = \frac 1{(1-F(v))^\lambda} h_F(v)$. 
\end{proposition}

\factone*
 \begin{proof}[Proof of \Cref{lem:1}]
Given $n$ i.i.d random variables drawn according to $F$,  the distribution of the first order statistic is given by $F^n(x)$. Thus, we want to show that if $F$ is $\lambda$- regular then $F^{n}$ is $\lambda$-regular. From \citet{schweizer2019performance} Proposition 1, this is equivalent to showing that $g_\lambda\circ F^n(x)$ is a convex function where
$g_\lambda(x) = \frac{1}{\lambda} \left( \frac{1}{(1-x)^{\lambda}} - 1 \right) $.

Assume first that $F$ is twice-differentiable (at the end of the proof, we show how to remove this assumption).  Under the differentiablity assumption, the previous condition turns to  $(g_\lambda \circ F^n)'' \geq 0$. 

Because, 
\begin{align*}
\frac{1}{n} (g_\lambda \circ F^n)'' &= f(x)^2\left(g_\lambda''(F^n(x)) n  F^{n-1}(x)^2 +  g_\lambda'(F^n(x)) (n-1) F^{n-2}(x)  \right) \\
                           &\quad  + g_\lambda'(F^n(x)) F^{n-1}(x)  f'(x),\\
                           g_\lambda'(x)  &= \frac{1}{(1-x)^{1+\lambda}},\\
                           g_\lambda''(x)  &=\frac{1+\lambda}{1-x} \cdot  g_\lambda'(x),
\end{align*}
 we have that $g_\lambda \circ F^n$ is convex if and only if
 $$ g_\lambda'(F^n(x)) \cdot F^{n-2}(x) \cdot \left[f(x)^2\left( \frac{1+\lambda}{1-F^n(x)} n  F^{n}(x)  + (n-1)\right)  + F(x) f'(x) \right]\geq 0.$$
 
 Using~\Cref{lem:lamb-aux}, whose proof is below, we obtain that
 \begin{align*}
\frac{1}{n} (g_\lambda \circ F^n)''    &\ge g_\lambda'(F^n(x)) \cdot F^{n-2}(x) \cdot \left[ \frac{1+\lambda}{1-F(x)}F(x) f^2(x)  + F(x)  f'(x) \right] \\
   &= F^{n-1}(x)  \cdot g_\lambda'(F^n(x))  \cdot \left[ \frac{1+\lambda}{1-F(x)}f^2(x)  +f'(x) \right] \\
   &= F^{n-1}(x)  \cdot g_\lambda'(F^n(x))  \cdot \left[ \frac{1+\lambda}{1-F(x)}f^2(x)  +f'(x) \right].
\end{align*}

Observe that in the last expression, the first term is trivially positive; the second term is positive since $g_\lambda'>0$; the third term is positive because $F$ is $\lambda$-regular.\footnote{A distribution twice-differentiable $F$ is $\lambda$-regular if and only if $(1-F(x))^{-\lambda}$ is convex, which is equivalent to $\frac{1+\lambda}{1-F(x)}f^2(x)  +f'(x) \ge 0$ for every $x$.} Therefore,  $g_\lambda \circ F^n$ is convex which implies that $F^n$ is a $\lambda$-regular distribution.

To finish the proof we need to tackle the case where $F$ is not necessarily twice-differentiable.  \Cref{cl:approx}, which is proven below, guarantees a sequence of distributions $(F_t)_{t\geq 0 }$ that are $\lambda$-regular and twice-differentiable satisfying that $F_t \to F$ uniformly on compact sets. Because $F_t$ are twice-differentiable, we have that  $g_\lambda \circ F_t^n$ is convex. Hence, taking the limit as $t\to \infty$, we conclude that $g_\lambda \circ F^n$ is convex. Therefore, $F^n$ is a $\lambda$-regular distribution.    
\end{proof}

\begin{lemma}\label{lem:lamb-aux}
Given a real $F \in [0,1]$ and $\lambda \in [0,1]$, the following inequality holds
 \begin{equation}\label{ineq:lem8} \frac{1+\lambda}{1-F^n} n F^{n} + (n-1) \geq \frac{1+\lambda}{1-F}  F 
 \end{equation}
 for every $n \in \mathbb{N}$. 
\end{lemma}
\begin{proof}

First, we assert that it is sufficient to show Inequality~\eqref{ineq:lem8} for the case $\lambda =1$. To see this, observe that Inequality~\eqref{ineq:lem8} is equivalent to 
$$ \frac{n F^{n}}{1-F^n} + \frac{n-1}{1+\lambda} \geq \frac{F}{1-F}  .$$  We conclude by observing that the lowest value of the left-hand-side happens when $\lambda=1$ while the right-hand-side is independent of $\lambda$.

Next, consider the function $h(F) = \frac{ n F^{n}}{ {1-F^n}} - \frac{F}{1-F} $. From the previous paragraph, we need to show that $h(F)\geq-\frac{n-1}{2}$. To this extent, we show that $h$ is decreasing and that $\lim_{F \to 1} h(F) = -\frac{n-1}{2}$. 

To tackle that $h$ is decreasing, observe that
\begin{align*}
   h'(F) &= \frac{n^2 F^{n-1}}{1-F^n}  + \frac{n^2 F^{2n-1}}{(1-F^n)^2} - \frac{1}{(1-F)} - \frac{F}{(1-F)^2} \\
        &= \frac{n^2 F^{n-1}}{(1-F^n)^2} - \frac{1}{(1-F)^2}.
        \end{align*}
Thus,
      $$ (1- F)^2 h'(F) = \frac{n^2 F^{n-1}}{(\sum_{i=0}^{n-1} F^{i})^2} - 1 .$$
      Because $F^{i} \geq F^{n-1}$ for $i\leq n-1$, we conclude that $(1 - F)^2 h'(F) \leq 0$. Therefore, $h$ is decreasing on $[0,1]$.

For the last step of the proof, observe that 
\begin{align*}
    \lim_{F \to 1} h(F) &= \lim_{F \to 1} \frac{n F^n}{1-F^n} - \frac{F}{1-F} \\
    &= \lim_{F \to 1} \frac{n F^n (1-F) - (1-F^n) \cdot F }{(1-F^n) (1-F)}  \\
    \intertext{Applying L'Hopital's rule twice, we get }
   \lim_{F \to 1} h(F)  &= \frac{n^2(n-1) - (n^2-1) \cdot n }{n(n+1) - n+1} \\
    & = -\frac{n-1}{2} .
\end{align*}
\end{proof}

\begin{lemma}\label{cl:approx}
Given a $\lambda$-regular distribution $F$, a sequence of distributions $(F_t)_{t\geq 0 }$ exists that are $\lambda$-regular and twice-differentiable such that $F_t \to F$ uniformly on compact sets.
\end{lemma}
\begin{proof}
 Proposition~\ref{prop:schw} tells us that $1-F(v)=\Gamma_\lambda(H_\lambda(r_\lambda(v)))$ with $H_v$ convex and increasing. \cite{russell1989representative} shows that a sequence of infinite differentiable convex increasing functions $(H_t(v))_{t \geq 0}$ exists such that 
$H_t\to H_\lambda$ uniformly on compact sets. Furthermore, we can construct $H_t$ so that $H_t(0)=H_\lambda(0)=0$ and $\lim_{v\to \infty}H_t(v)=\infty$. Hence, the distribution $F_t(v) = 1- \Gamma_\lambda(H_t(v))$ is well-defined, twice-differentiable, satisfies that $F_t\to F$ uniformly on compact sets ($\Gamma_\lambda$ is a continuous function), and from Proposition~\ref{prop:schw} is $\lambda$-regular.
\end{proof}

\section{Missing proofs from~\Cref{sec:model}}

\lemzero*
\begin{proof}[Proof of Lemma~\ref{lem:0}]
Given a BIC $(\vx,\vp)$, the IR constraint and the Condition 2 on the utility function $U_\ag$ imply that $\E[p_\ag(\vv_\ag,\vv_{-\ag}) | \vv_\ag] \le \E[\sum_{i \in \pi(\ag)}x_{\bu}(\vv_i,\vv_{-i}) v_i]$. Therefore, 
$$\mbox{Rev}(F,m,\pi;(\vx,\vp))= \sum_{\ag\in[m]} \E[p_\ag(\vv)] \le \sum_{\ag\in[m]}  \E\left [\sum_{i=1}^{x_{\ag}} \ostat{i}{\ag }\right] \le \mbox{Wel} (F).$$
We conclude that $\mbox{Rev}(F,m,\pi)\le \mbox{Wel}(F)$.
\end{proof}

%

\section{Missing proofs from~\Cref{sec:robustness}}
\propone*
\begin{proof}[Proof of Proposition~\ref{prop:1}]
We show a more general result which has Proposition~\ref{prop:1} as a specific case. We claim that the profit-maximizing mechanism for selling $k$ objects, each having a production cost $c$, to $n$ bidders consists on setting a $(k+1)$-price auction with reserve price $\phi^{-1}_F(c)$.

The proof consists of adapting the technique used in \citet{myerson1981optimal} for this particular setting. First, let $U_j(x_j,p_j|v_j) = v_j x_j - p_j$ the bidder's utility. Given $(\vx,\vp)$ a direct mechanism and we denote by $W_{j}(v_j) = \E[U_j(x_j(v_j,\vv_{-j}),p_j(v_j,\vv_{-j}))|v_j]$. The quasi-linearity of $U_j$ implies that the incentive compatibility of $(\vx,\vp)$ is equivalent to (i) that $\chi^{x_j}_j(v_j) = \E_{\vv_{-j}}[x_j(v_j,\vv_{-j})]$ is non-decreasing, and $W_{j}(v_j) = W_{j}(\underline v) + \int_{\underline v}^{v_j}\chi_j(z)dz$ \citep{milgrom2002envelope}. The (IR) constraint implies that $W_{j}(\underline v) = 0$. Furthermore, the equivalence result pins-down the price $p_j$ as function of the allocation rule $x_j$ by : $\E_{\vv_{-j}}[p_j(v_j,\vv_{-j})] = v_j \chi_j(v_j) -  W_j(v_j)$.

Using the previous characterization, the seller's problem turns to 
\begin{align*}
   \max_{\vx:[\underline v,\overline v]^n\to [0,1]^n}  & \quad \sum_{j\in [n]} \E_{v_j}\left[ v_j \chi^{x_j}_j(v_j) -  V_j(v_j) - c \chi^{x_j}_j(v_j)  \right]  = \E_{\vv}\bigg[\sum_{j\in [n]} (\phi_{F}(v_j) -c) x_j(\vv)\bigg]\\
   \mbox{ subject to } &\quad \chi^{x_j}_j  \mbox{is non-decreasing for every } j\in [n] \\
   & \quad \sum_{j\in [n]} x_j(\vv) \le k,
\end{align*} 
where the equality in the optimization problem comes from integration by parts.

Solving the relaxed problem without taking into account the monotonicity constraint, we obtain that the solution is 
$$x_j^*(\vv) = \begin{cases}
 1 &\mbox{ if } v_j \ge v^{(k)} \mbox{ and }\phi_F(v_j)\ge c \\
0 &\mbox{ otherwise}
\end{cases}.$$
Accordingly, $p^*_j(\vv) = x_j^*(\vv) \max\{ v^{(k+1)}, \phi^{-1}_F(c) \}$ is an optimal pricing scheme.

Finally, using that $F$ has the MHR property we conclude that $\chi_j^{x^*_j}$ is non-decreasing. Therefore, $(\vx^*,\vp^*)$ is the revenue-maximizing mechanism. 
\end{proof}

%

\proptwo*
\begin{proof}
Let $k=1$. First note that, the revenue of the optimal mechanism with monopsony ($\mbox{Rev}^M$) is at least as much as the mechanism that sells the item with a posted price $p = \E[\ostat{1}{n}]$. Hence, we obtain that
$$\mbox{Rev}^M \ge \E[\ostat{1}{n}] \;\P[\ostat{1}{n} \ge \E[\ostat{1}{n}]]  \geq \frac 1 e \E[\ostat{1}{n}],$$
where the second inequality is due to Lemma~\ref{lem:2}. On the other hand, as noted in Proposition~\ref{prop:1} the optimal competition mechanism is a second-price auction with a reserve price. In the monopsony case, we have that in every Bayes Nash Equilibrium the intermediary will only buy the item, for a price $r$, when $\ostat{1}{n}\ge r$. Thus the revenue of the optimal competition mechanism in a market with a monopsony demand structure is $\mbox{Rev}^M(\vx^C,\vp^C) = r \P[\ostat{1}{n} \ge r] \le r$.  

Next, let $F$ be standard exponential distribution with $f(x) = exp(-x)$ and $F(x) = 1 - exp(-x)$. The hazard rate $h_F(z) = 1$ which implies that the reserve price is $r=1$ (independent of $n$). For exponential distributions, $\E[\ostat{1}{n}] = \sum _{j=1}^{n}{\frac {1}{n-j+1}} = \mathcal{H}(n)$ \cite{renyi}. Thus $\lim_{ n \to \infty} \E[\ostat{1}{n}] = \infty$, and we conclude that $\lim_{n \to \infty} \frac {\mbox{Rev}^M(\vx^C,\vp^C)} {\mbox{Rev}^M} = 0$.
\end{proof}

\propmonopsony*
\begin{proof}
The revenue of setting the optimal pricing $p_B$ is greater or equal than the revenue of selling the bundle at price $\E[\sum_{\ob=1}^k \ostat{\ob}{n}]$. Because $G$ satisfies the MHR property (combine~\Cref{lem:1} with the fact MHR distributions are closed under summation~\cite{barlow1963properties}), Lemma~\ref{lem:2} implies that
$$ \E\left[\sum_{\ob=1}^{k} \ostat{\ob}{n}\right] \cdot \P\left[\sum_{\ob=1}^{k} \ostat{\ob}{n} \ge \E\left[\sum_{\ob=1}^{k} \ostat{\ob}{n}\right] \right] \ge \frac {1}e  \; \E\left[\sum_{\ob=1}^{k} \ostat{\ob}{n} \right] = \frac 1 e \; \mbox{Wel}(F).$$

Using Lemma~\ref{lem:0}, we conclude that pricing the bundle at $p_B$ obtains a revenue of at least $\frac 1 e$ of Rev$^M$.

For $k=n$, the Law of Large Numbers
\footnote{Since the distribution F is MHR, the moments are bounded and law of large numbers can be applied \citep{barlow1963properties}.}
implies that $\lim_{n\to \infty}\P[\sum_{\ob=1}^{n} \ostat{\ob}{n} \ge p_B^\epsilon] =1$ for every $\epsilon >0$. Hence,
$ \lim_{n\to \infty} \frac{p_B^\epsilon \P[\sum_{\ob=1}^{n} \ostat{\ob}{n} \ge p_B^\epsilon]}{\mbox{Wel}(F)} = 1- \frac{\epsilon}{\E[v]}.$
We conclude by taking $\epsilon \to 0$.
\end{proof}

\corgoodmonopsony*
\begin{proof}
Consider the case $k=n$. From Proposition~\ref{prop:2} we have that the bundle price $p_B^\epsilon$ is asymptotically optimal. Hence, the revenue of the bundle is approximately $ n \E[v]$, while the revenue of Competition is $n r \P[v\ge r]$. From Markov Inequality, we know that $\E[v] = \E[v| v\ge r] \P[v \ge r] + \E[v|v < r] + \P[v < r] \ge r \P[v\ge r]$.  Further the last inequality is strict as long as $\P[v > r] > 0$ or $ \E[v| v < r] > 0$ and $\P(v < r) > 0$, which is true for distribution with at least two non-zero points in its support.  Thus, the revenue of the bundle is more than the revenue from Competition.
\end{proof}

\propbundlebad*
\begin{proof}
For the competition case the probability of selling the bundle is equivalent to the probability that the highest valuation bidder would like to get the bundle. Hence, $Rev^C(\vx^B,\vp^B) = p^\epsilon_B \P[\ostat{1}{n}\ge p^\epsilon_B] \le  n \E[v]\P[\ostat{1}{n}\ge  n (\E[v]-\epsilon) ]$. On the other hand, Rev$^C= n r \P[v\ge r]$. Therefore, we have that
$$ \frac{\mbox{Rev}^C(\vx^B,\vp^B)}{ \mbox{Rev}^C} \le \frac{\E[v] \cdot \P[\ostat{1}{n}\ge n(\E[v]-\epsilon)]} {r\P[v \ge r]} \le \frac{\E[v] \cdot E[\ostat{1}{n}]}{n (\E[v]-\epsilon) \cdot r P[v \ge r]},$$ 
where the second inequality comes from Markov Inequality. We conclude that for a bounded distribution $F$, $\lim_{n \to \infty} \frac{\mbox{Rev}^C(\vx^B,\vp^B)}{ \mbox{Rev}^C} = 0$.
\end{proof}

\section{Missing proofs from~\Cref{sec:robustmechanism}}

\lemdmone*
\begin{proof}[Proof of Lemma~\ref{lem:dml}]
    The optimal behavior of the monopolist, discussed in Lemma~\ref{lem:dm}, implies that what we just need to show that $\P[\phi_F(v_\bu)\ge p | v_\bu \ge p] \ge (1-\lambda)^{\frac 1 \lambda}$.
    
    From Bayes' rule we have that 
    \begin{align*}\P[\phi_F(v_\bu)\ge p | v_\bu \ge p] &= \frac{\P[\phi_F(v_\bu)\ge p , v_\bu \ge p]}{\P[ v_\bu \ge p]} =  \frac{\P[\phi_F(v_\bu)\ge p]}{\P[ v_\bu \ge p]}\\&=\frac{1-F(\phi^{-1}_F(p))}{1-F(p)},
    \end{align*}
    where the second equality follows from the fact that $\phi_F(v_\bu)  = v_\bu - (1 - F(v_\bu))/f(v_\bu) \le v_\bu$.
    
    Because $F$ is $\lambda$-regular, Proposition~\ref{prop:schw} implies that that $1-F(v)= \Gamma_\lambda(H_\lambda(v))$ where $\Gamma_\lambda(v) = (1+\lambda v)^{-\frac 1 \lambda}$ and that $H_\lambda(v) = \int_0^v r_\lambda(z) dz$ where $r_\lambda(v) = \frac 1{(1-F(v))^\lambda} h_F(v)$. Therefore,
    \begin{align}\P[\phi_F(v_\bu)\ge p | v_\bu \ge p] &=\frac{\Gamma_\lambda(H_\lambda(\phi^{-1}_F(p)))}{\Gamma_\lambda(H_\lambda(p))}. \label{eq:andres}
    \end{align}
    
    To tackle the right-hand-side of the above expression, we derive the following inequalities. 
    
    First, because  $F$ is $\lambda$-regular, we have that $H_\lambda$ is convex with $H_\lambda'(v) =r_\lambda(v)$ (\Cref{prop:schw}). The convexity and differentiability of $H_\lambda$ implies that \begin{align*} 
    H_\lambda(\phi_F(p)) &\ge H_\lambda(p) + r_\lambda(p) (\phi_F(p) - p )\\
    &= H_\lambda(p) + \frac 1{(1-F(p))^\lambda}  h_F(p) (\phi_F(p) - p ).
    \end{align*}
    
    Second, from the definition of the virtual valuation we have that $\phi_F(p) = p - \frac 1 {h_F(p)}$. Replacing in the previous inequality we conclude that  $H_\lambda(\phi_F(p))\ge H_\lambda(p) - (1-F(p))^{-\lambda}$. Moreover, noticing that $(1-F(p))^{-\lambda}= \Gamma_\lambda (H_\lambda(p))^{-\lambda} = 1+\lambda H_\lambda(v)$, we obtain that
    $$ H_\lambda(\phi_F(p))\ge H_\lambda(p) (1-\lambda) -1.$$
    
    Third, because $\Gamma_\lambda$ is decreasing and that $\phi_F$ is monotone,  we derive from the previous inequality that
    $$ \Gamma_\lambda (H_\lambda (p)) \le \Gamma_\lambda\big(\; H_\lambda(\phi^{-1}_F(p)) (1-\lambda) -1\; \big).$$
    
    To conclude the proof, we replace the above expression in Equation~\eqref{eq:andres} and obtain that
    \begin{align*}
       \P[\phi_F(v_\bu)\ge p | v_\bu \ge p] &\ge \frac{\Gamma_\lambda(H_\lambda(\phi^{-1}_F(p)))}{\Gamma_\lambda\big(H_\lambda(\phi^{-1}_F(p)) (1-\lambda) -1 \big)}\\
       &= \left( \frac {1+\lambda H_\lambda(\phi^{-1}_F(p))}{1+\lambda \left( H_\lambda(\phi^{-1}_F(p)) (1-\lambda) -1\right)}  \right)^{-\frac 1 \lambda}\\
       & = \left({1-\lambda}\right)^{\frac 1 \lambda}.
    \end{align*}
 \end{proof} 
 
\section{Missing proofs from~\Cref{sec:ext}}

\begin{lemma}\label{cl:1}
For $ \ob \geq 1$,
$$n \cdot \sum_{t \leq j} p_t \geq  \frac{\ob}{2}(1-\lambda)^{\frac 1 \lambda}.$$  
\end{lemma}
\begin{proof}

Observe that when $\ob =1$, 
$$np_j =n\P[x \geq \E[\ostat{1}{n}]] \geq \P[\ostat{1}{n} \geq \E[\ostat{1}{n}]] \geq (1-\lambda)^{\frac 1 \lambda}.$$ Here the first inequality follows from union bound and the second inequality follows from~\Cref{lem:2}.
For $\ob > 1$, 
$n \sum_{t \leq j} p_j = n \P[x \geq u_j]$. Recall that $u_j = E[\ostat{1}{\ceil{n/\ob}}]$.
We can partition the set of $n$ draws into $\ob$ sets each containing $s= \ceil{\frac{n}{\ob}}-1$ values. From~\Cref{lem:1}, we know that the maximum element in each partition will lie in the bucket with probability at least $\frac{1}{2}(1-\lambda)^{\frac 1 \lambda}$ (as $l \geq 2$). Therefore,
$$n \P[x \geq u_j] \ge j \cdot (s) \P[x \geq u_j] \geq j \cdot \P[\ostat{1}{s} \geq \E[\ostat{1}{s+1}] ] \geq j/2 (1-\lambda)^{\frac 1 \lambda}.$$
\end{proof}

\begin{lemma} \label{lem:optprog}
 Consider $r_1 \geq r_2 \geq \cdots \geq r_k \geq 0$,  then an optimal solution to the convex program \begin{align*}
   \max \sum_{\ob=1}^k r_\ob &\cdot    e^{-p_\ob \cdot n}  \label{prog:P1} \tag{Primal}\\
   n \cdot \left(\sum_{j=1}^s p_j \right) &\geq \frac{s}{2}  (1-\lambda)^{\frac 1 \lambda} \qquad \text{ for all } 1\leq s \leq k \\
   0 \leq p_\ob &\leq 1  \qquad\qquad \qquad  \text{ for all } 1\leq \ob \leq k
\end{align*}
 is $p_1 = p_2 = \dots =  p_k = \frac{1}{2n} (1-\lambda)^{\frac 1 \lambda}$.
\end{lemma}
\begin{proof}
Our proof-technique uses the complementary slackness result for convex duality. 

Given the Primal problem, we denote by $\lambda_s$ the dual variable corresponding to the first constraint and define $S_\ob = \sum_{s=\ob}^k \lambda_s$. Naturally $S_\ob \geq S_{\ob + 1}$. If we lagrangify the first constraint and find the optimal value of $p_\ob$, we get that $p_\ob = -\log(\frac{nS_\ob}{r_\ob})/n$. Substituting this value gives us the objective. The constraint $0 \leq p_\ob \leq 1$ gives us the second constraint on $S_j$. Thus we obtain the following dual problem.
\begin{align*}
   \min \sum_{\ob=1}^k S_\ob \left( n - \frac{1}{2}(1-\lambda)^{\frac 1 \lambda} \right)  &- S_\ob \cdot \log \left(\frac{nS_\ob}{r_\ob} \right)  \label{prog:D1} \tag{Dual}\\
    S_\ob &\geq  S_{\ob+1}  \qquad \quad \text{ for all } 1\leq \ob \leq k-1 \\
   \frac{r_\ob}{ne^n} \leq S_\ob &\leq  \frac{r_\ob} {n} \qquad  \qquad \text{ for all } 1\leq \ob \leq k.
\end{align*}

Consider the primal solution $p_\ob = \frac{1}{2n}(1-\lambda)^{\frac 1 \lambda}$ and observe that this is a feasible solution and achieves an objective value of $\sum_{j=1}^k r_j e^{-1/2(1-\lambda)^{\frac 1 \lambda} }$.  Consider setting the dual to be $S_\ob = r_\ob \cdot e^{-1/2 (1-\lambda)^{\frac 1 \lambda}} / n$. Since $r_1 \geq \dots r_k \geq 0$, we know that this dual is feasible and has the same objective value. Since both the primal and dual are feasible and bounded,  we know that the primal solution $p_\ob$ is optimal. 
\end{proof}

\section{Missing proofs from Section~\ref{sec:discussion}} 

\lemregular*
\begin{proof}
Let $n > e^4$ be an integer. Consider a distribution with support $[1, n)$ with density $f(x) = \frac{n}{n-1} \times \frac{1}{x^2}$. Note that for this distribution the C.D.F $F(x) = \frac{n}{n-1}\times (1 - 1/x)$. The optimal revenue for selling single item with values drawn from distribution $f$ is $\arg \max_x x ( 1- F(x)) = \arg \max x (1/x - 1/n) \frac{n}{n-1} = (1 - 1/n) \frac{n}{n-1} = 1$. Thus the optimal revenue for selling $n$ items to $n$ I.I.D. bidders with values drawn from $f$ is $n$. In the following lemma, we show that when the demand structure is monopsony, it is possible to use bundle pricing to get $\Omega(n \ln{n})$ revenue.

We will show that the expected revenue from setting a bundle price of $n \ln{n}/2$ is $\Omega(n \ln{n})$. For a bundle price $p$, the intermediary will purchase the bundle if the sum of the $n$ values of its buyers is at least $p$. Let $X_1, X_2, \ldots X_n$ denote the random variable indicating the buyers draws and let $X = X_1 + X_2 + \cdots + X_n$, denote the random variable for the sum of these $n$ random draws. For each $X_i$, $\E[X_i] = \frac{n}{n-1} \ln{n}$, $\E[X_i^2] = n$, and hence $Var[X_i] = n - (\frac{n \ln{n}}{n-1})^2$. Further $\E[X] = \frac{n^2 \ln{n}}{n-1}$ and $Var[X] = n^2 - \frac{n^3 \ln{n}^2}{(n-1)^2} < n^2$. 

We will use the following one-sided Chebychev's inequality to prove this result. For any random variable with mean $\mu$ and standard deviation $\sigma$,
$$ \P[|x - \mu| \geq k \sigma] \leq 1/k^2 \implies \P[\mu - k \bar{\sigma} \geq x] \leq 1/k^2, $$ 
for any $\bar{\sigma} \ge \sigma$. 
Setting $\mu = \frac{n^2 \ln{n}}{n-1}$, $\bar{\sigma} = n$, and $k = \frac{\mu}{2\bar{\sigma}}$, 
$$ \P[X \leq \mu - k \bar{\sigma} = \mu/2] \leq 1/k^2 = \frac{4 \bar{\sigma}^2}{\mu^2} = \frac{4n^2 (n-1)^2}{n^4 \ln{n}^2} \leq 1/4. $$ 
In the last step we use that $n > e^4$.

Thus $\P[X \geq \E[X]/2] \geq 3/4$ and setting a bundle price of $\E[X]/2$, will yield a revenue of $\Omega(n \ln {n})$.

On the other hand, with item pricing, the most revenue that can be obtained is $O(n)$. Thus item-pricing obtain at most $\Omega(\ln n)$ factor to the optimal revenue.
\end{proof}

\end{document}